\documentclass{amsart}
\usepackage{mathrsfs}
\usepackage{amsfonts}
\usepackage{amssymb, verbatim, amsthm, amsfonts, amsmath}
\usepackage{graphics,color}
\usepackage[all]{xy}

\usepackage{graphicx}
\usepackage{float}
\usepackage{subfigure}

\usepackage[ruled,linesnumbered]{algorithm2e}
\usepackage{algpseudocode}
\usepackage{amsmath}
\algdef{SE}[DOWHILE]{Do}{doWhile}{\algorithmicdo}[1]{\algorithmicwhile\ #1}%

\usepackage{url}
\usepackage{tikz}
\usetikzlibrary{arrows}

\newtheorem{theorem}{Theorem}[section]
\newtheorem{lemma}[theorem]{Lemma}

\newtheorem{example}{Example}[section]

\newcommand{\edit}[1]{{\color{red}(#1)}}

\newcommand{\etal}{\emph{et al.~}}
\def\bs{{\mathbf{s}}}
\newcommand{\ie}{\emph{i.e.}}
\def\F{{\mathbb{F}}}
\newcommand{\eg}{\emph{e.g.}}

\begin{document}

\title[]{Efficient generation of odd order de Bruijn sequence with the same complement and reverse sequences}

\author{Zuling Chang and Qiang Wang}

\address{Z. Chang,  School of Mathematics and Statistics, Zhengzhou University, Zhengzhou 450001, China}
\email{zuling\textunderscore chang@zzu.edu.cn}

\address{Q. Wang,  School of Mathematics and Statistics, Carleton University, 1125 Colonel By Drive, Ottawa, Ontario, K1S 5B6, Canada}
\email{wang@math.carleton.ca}

\begin{abstract}
Experimental results show that, when the order  $n$ is odd, there are de Bruijn sequences such that the corresponding complement sequence and the reverse sequence are the same.
In this paper, we propose one efficient method to generate such de Bruijn sequences.  This  solves an open problem asked by
Fredricksen  forty years ago for showing  the existence of such de Bruijn sequences when the odd  order $n >1$.  Moreover,  we refine a characterization of  de Bruijn sequences 
with the same complement and reverse sequences and study the number of these de Bruijn sequences, as well as the distribution of  de Bruijn sequences of the maximum linear complexity.
\end{abstract}

\keywords{de Bruijn sequence, complement sequence, reverse sequence, cycle joining method}

\maketitle

\section{Introduction}
A $2^n$-periodic binary sequence is a binary de Bruijn sequence of order $n$ if every binary $n$-tuple occurs exactly once within each period. There are $2^{2^{n-1}-n}$ such sequences~\cite{Bruijn46}. These
de Bruijn sequences have many good properties including balance, maximum period and large linear complexity and they are applicable in coding, communication and cryptography ~\cite{Chan82,Golomb,HH96}.

For a given binary $N$ periodic sequence ${\bf s}=(s_0,s_1,\ldots,s_{N-1})$, its complement sequence is $\textbf{\emph{c}}{\bf s}=(\overline{s_0},\overline{s_1},\ldots,\overline{s_{N-1}})$ where $\overline{s_i}=1+s_i ~({\rm mod}~ 2)$, and its reverse sequence is $\textbf{\emph{r}}{\bf s}=(s_{N-1},s_{N-2},\ldots,s_0)$. For a given de Bruijn sequence ${\bf s}$, it is proved in \cite{Etzion84,Fred82,Mayhew94} that ${\bf s}$ and $\textbf{\emph{c}}{\bf s}$ are inequivalent (the one isn't the cyclic shift of the other), ${\bf s}$ and $\textbf{\emph{r}}{\bf s}$ are inequivalent, and if the order is even then $\textbf{\emph{c}}{\bf s}$ and $\textbf{\emph{r}}{\bf s}$ are inequivalent. From these facts we know when the order is even the number of de Bruijn sequences with given linear complexity  is divided by 4 \cite{Etzion84}.
Experimental results show that, when the order is odd, there are de Bruijn sequences satisfying $\textbf{\emph{c}}{\bf s}=\textbf{\emph{r}}{\bf s}$, {\it i.e.}, ${\bf s}=\textbf{\emph{cr}}{\bf s}$. In \cite{Etzion84}, the sequence ${\bf s}$ with ${\bf s}=\textbf{\emph{cr}}{\bf s}$ is called $CR$-sequence.  In this paper, the de Bruijn sequence being $CR$-sequence is called {\it CR de Bruijn sequence.} These sequences help understanding the hidden nature of de Bruijn sequences, for example, the distribution of  de Bruijn sequences.
 On page 218 in \cite{Fred82}, Fredricksen wrote, ``A harder problem would be to show that there are examples like this whenever $n$ is odd $>1$.''
 In this paper we solve this problem of Fredicksen by providing a constructive method to generate de Bruijn sequence satisfying ${\bf s}=\textbf{\emph{cr}}{\bf s}$ for arbitrary odd order. 
 In order to explain the main idea of our approach, we need a few notations and known results.

An {\it $n$-stage shift register} is a circuit of $n$ consecutive storage units, each containing a bit. The circuit is clock-regulated, shifting the bit in each unit to the next stage as the clock pulses. A shift register generates a binary code if one adds a feedback loop that outputs a new bit $s_n$ based on the $n$ bits $ s_0,\ldots,s_{n-1}={\bf s}_0\in\mathbb{F}_2^n$, which is called an initial state of the register. The corresponding Boolean {\it feedback function} $f(x_0,\ldots,x_{n-1})$ outputs $s_n$ on input ${\bf s}_0$. The output of a {\it feedback shift register} (FSR) is therefore a binary sequence ${\bf s}=\{s_i\}=s_0,s_1,\ldots,s_n,$ $\ldots$ that satisfies the recursive relation
\[
s_{n+\ell} = f(s_{\ell},s_{\ell+1},\ldots,s_{\ell+n-1}) \text{ for } \ell = 0,1,2,\ldots.
\]

For $N \in \mathbb{N}$, if $s_{i+N}=s_i$ for all $i \geq 0$, then ${\bf s}$ is {\it $N$-periodic}
or {with period $N$} and one writes ${\bf s}= (s_0,s_1,s_2,\ldots,s_{N-1})$. The least among all periods of ${\bf s}$ is called the {\it least period} of ${\bf s}$.
For any FSR, distinct initial states generate distinct sequences and thus there are $2^n$ distinct sequences generated from an FSR with feedback function $f(x_0,x_1,\ldots,x_{n-1})$. All these sequences are periodic if and only if $f$ is {\it nonsingular}, \ie, $f$ can be written as
\[
f(x_0,x_1,\ldots,x_{n-1})=x_0+h(x_1,\ldots,x_{n-1}),
\]
for some Boolean function $h(x_1,\ldots,x_{n-1})$ whose domain is $\F_2^{n-1}$~\cite[p.~116]{Golomb}.

We name ${\bf s}_i= s_i,s_{i+1},\ldots,s_{i+n-1}$ the {\it $i$-th state} of sequence ${\bf s}$ and we say that ${\bf s}_i$ is in $\bs$. Its {\it predecessor} is ${\bf s}_{i-1}$ while its {\it successor} is ${\bf s}_{i+1}$. Extending the definition to any binary vector or sequence ${\bf s}= s_0,s_1,\ldots,s_{n-1}, \ldots$, we let
$\overline{\bf s} :=
\overline{s_0},\overline{s_1},\ldots,\overline{s_{n-1}}, \ldots$. An arbitrary state ${\bf v}=v_0,v_1,\ldots,v_{n-1}$ of ${\bf s}$ has
\[
\widehat{{\bf v}}:=\overline{v_0},v_1,\ldots,v_{n-1} \mbox{ and }
\widetilde{{\bf v}}:=v_0,\ldots,v_{n-2},\overline{v_{n-1}}
\]
as its {\it conjugate} state and {\it companion} state, respectively. Hence, $({\bf v}, \widehat{{\bf v}})$ is a {\it conjugate pair} and $({\bf v}, \widetilde{{\bf v}})$ is a {\it companion pair}.

For an $N$-periodic sequence $\mathbf{s}$, the {left shift operator} $L$ maps $(s_0,s_1,\ldots,s_{N-1})$ to  $(s_1,s_2,\ldots,s_{N-1},s_0)$, with the convention that $L^0$ fixes $\mathbf{s}$. The set
\begin{equation*}
[\mathbf{s}]=\left\{\mathbf{s},L\mathbf{s},\ldots,L^{N-1}\mathbf{s} \right\}
\end{equation*}
is a {\it shift equivalent class}. Sequences in the same shift equivalent class correspond to the same cycle in the state diagram of FSR \cite{GG05}.
We call a periodic sequence in a shift equivalent class a {\it cycle}. If a nonsingular FSR with feedback function $f$ generates $r$ disjoint cycles
$C_1, C_2, \ldots, C_r$, then its {\it cycle structure} is
\[
\Omega(f)=\{C_1, C_2, \ldots, C_r\}.
\]

The {weight} of an $N$-periodic cycle $C$, denoted by $wt(C)$, is
\[
|\{ 0 \leq j \leq N-1 : c_j = 1  \}|.
\]
Similarly, the weight of a state is the number of $1$s in the state. The lexicographically least $N$-stage state in any $N$-periodic cycle is called its {\it necklace}. As discussed in, \eg, \cite{Booth80} and~\cite{Gabric18}, there is a fast algorithm that determines whether or not a state is a necklace in $O(N)$ time.

A description of the cycle joining method (CJM) \cite{Fred82,Golomb} is as follows.
Given disjoint cycles $C$ and $C'$ in $\Omega(f)$ with the property that
some state $\mathbf{v}=v_0,v_1,\ldots,v_{n-1}$ in $C$ has its conjugate state $\widehat{\mathbf{v}}$ in $C'$, interchanging the successors of $\mathbf{v}$ and $\widehat{\mathbf{v}}$ joins $C$ and $C'$ into a cycle whose feedback function is
\begin{equation}\label{eq:newfeedback}
\widehat{f}:=f(x_0,x_1,\ldots,x_{n-1})+\prod_{i=1}^{n-1}(x_i+\overline{v_{i}}).
\end{equation}
Similarly, if the companion states $\mathbf{v}$ and $\widetilde{\mathbf{v}}$ are in two distinct cycles, then interchanging their predecessors joins the two cycles. If this process can be continued until all cycles that form $\Omega(f)$ merge into a single cycle, then we obtain a de Bruijn sequence. The cycle joining method is, therefore, predicated upon knowing the cycle structure of $\Omega(f)$ and is closely related to a graph associated to the FSR.

Given an FSR with feedback function $f$, its {\it adjacency graph} $G_f$, or simply $G$ if $f$ is clear, is an undirected multigraph whose vertices correspond to the cycles of $\Omega(f)$. The number of edges between two vertices is the number of shared conjugate (or companion) pairs, with each edge labelled by a specific pair. It is well-known that there is a bijection between the set of spanning trees of $G$ and the set of all inequivalent de Bruijn sequences constructible by the CJM \cite{HH96}.


Gabric \etal ~\cite{Gabric18} replaced the usual CJM generating algorithm by  the successor rule $\rho(x_0,x_1,\ldots,x_{n-1})$, to generate the next bit. Given an FSR with feedback function $f(x_0,x_1,\ldots,x_{n-1})$, the general thinking behind this approach is to determine some condition which guarantees that the resulting sequence is de Bruijn.
Let $A$ be a set of conjugate pairs. For each state ${\bf c}=c_0,c_1,\ldots,c_{n-1}$, a nontrivial successor rule $\rho$ is an assignment
\begin{equation}\label{equ:rho1}
\rho(c_0,c_1,\ldots,c_{n-1}) :=
\begin{cases}
\overline{f(c_0,\ldots,c_{n-1})} & \mbox{if } {\bf c} \in A,\\
f(c_0,\ldots,c_{n-1}) & \mbox{if } {\bf c} \notin A.
\end{cases}
\end{equation}
To be precise, the successor of ${\bf c}=c_0,c_1,\ldots,c_{n-1}$ is $c_1,\ldots,c_{n-1}, f(c_0,\ldots,c_{n-1})$ except when ${\bf c} \in A$. When ${\bf c} \in A$, the successor is $c_1,\ldots,c_{n-1}, f(c_0,\ldots,c_{n-1})+1$.
To guarantee that the successor rule can generate a de Bruijn sequence, the set $A$ must be carefully selected to ensure that the conjugate pairs in it can be used to join all cycles into one. Several explicit choices of successor rules can be found in \cite{Sawada16, Zhu21}, which can also be viewed as special cases of \cite{Jansen91}.


A $CR$-sequence is characterized  by Etzion and Lempel in  \cite{Etzion84}. In the following $\textbf{\emph{cr}}X$ denotes the reverse of the complement of $X$.

\begin{lemma}\label{lemma1}\cite{Etzion84}
A sequence ${\bf s}$ is a $CR$-sequence if and only if the period is even and  ${\bf s}=(X\textbf{\emph{cr}}X)$, for some $X$. \end{lemma}

From this characterization,  in order to generate a CR de Bruijn sequence, we need to find a sequence $X$ with period $2^{n-1}$ such that a state ${\bf v}$ is in $X$  but $\textbf{\emph{cr}}{\bf v}$ isn't in $X$. Then we can use CJM to  join $X$ and $\textbf{\emph{cr}}X$ into one cycle by a suitable conjugate pair. However, Finding such an $X$ by checking all its states is a little complicated, hence we consider this problem by joining cycles of some simple FSRs directly.

The {\it Pure Cycling Register} (PCR) of order $n$ is an FSR with feedback function
\[
f(x_0,x_1,\ldots,x_{n-1})=x_0,
\]
Each cycle generated by PCR is $n$-periodic and has the form $(c_0,c_1,\ldots,c_{n-1})$. Let $\phi(\cdot)$ be the Euler totient function, it is well known, for example,  in \cite{Golomb,Sloane02},  that the number of cycles generated by PCR of order $n$ is
\begin{equation}\label{eq:number}
Z_{n} = {\frac{1}{n}} {\sum_{d|n}}\phi(d) \,2^{\frac{n}{d}}=|\Omega(f)|.
\end{equation}
This is also the number of irreducible polynomials with degrees dividing $n$ over the binary field. When $n>2$, $Z_n$ is even.

In order to generate an odd order de Bruijn sequence such that it is also a $CR$-sequence, we partition the cycle structure of PCR into two parts according to their weights. Then we join these cycles into two large cycles within each part respectively,  according to two distinct carefully selected successor rules. Finally we join two  large cycles into one de Bruijn sequence by selecting a conjugate pair of a special format.  More specifically, if the order of PCR is odd, then for a given PCR cycle $C=(c_0,c_1,\ldots,c_{n-1})$, one can check that $C'=(\overline{c_{n-1}},\overline{c_{n-2}},\ldots,\overline{c_{0}})$ is also a PCR cycle, $C\neq C'$ and $wt(C)+wt(C')=n$. Obviously there is a one-to-one correspondence between $C$ and $C'$.  In fact,  we can write $C'$ as $\textbf{\emph{cr}}C$ because $C'$ is obtained from $C$ after taking the complement and reverse. We divide $\Omega(f)$ of PCR with odd order into two subsets
\begin{equation}\label{eq:two}
T_0 =\{C\in\Omega(f) \mid wt(C)\leq \frac{n-1}{2} \} \ \mbox{and} \ T_1 =\{C\in\Omega(f) \mid wt(C)\geq \frac{n+1}{2} \}.
\end{equation}

Our approach is to  use CJM to join the cycles in $T_0$ and $T_1$ respectively according to different PCR successors.  This generates two large cycles $S_0$ and $S_1$ satisfying $S_0=\textbf{\emph{cr}}S_1$. After joining these two large cycles into one by a specially chosen conjugate pair, a de Bruijn sequence with odd order and $\textbf{\emph{c}}{\bf s}=\textbf{\emph{r}}{\bf s}$ will be generated.

The proof of main results will be presented  in Section~\ref{sec:main}. The correctness of results follows by the choices of conjugate pairs and the existence of the spanning trees in the adjacency graph.   Our construction can efficiently generate a lot of CR de Bruijn sequences. Morever,  we study the characterization and properties of such special de Bruijn sequences in Section \ref{Sec-properties}.  Using a refined characterization of these de Bruijn sequences,  we extend a result of Etzion and Lempel by showing that the number of order $n$ de Bruijn sequences of the maximum linear complexity is congruent to $0$ modulo $16$.  A short conclusion is  given in Section~\ref{Sec4}.


\section{The proof of the existence}\label{sec:main}

In this section we will show how to generate CR de Bruijn sequences in the form of successor rules. 
We divide the proof of  the main result into several lemmas. 
We assume $n$ is odd from now on.

The following result slightly modifies a well known PCR de Bruijn successor rule studied by Jansen \etal \cite{Jansen91} and Sawada \etal  \cite{Sawada16}, which is called {\it PCR3 successor} rule in \cite{Gabric18}.  We introduce the restriction on their cycle weights  in order to obtain a large cycle joining from all cycles of weights less than or equal to $\frac{n-1}{2}$.

\begin{lemma}\label{lem1}
For each state ${\bf c}=c_0,c_1,\ldots,c_{n-1}$ such that $wt({\bf c})\leq \frac{n-1}{2}$, we let ${\bf u}_{\bf c}=c_1,c_2,\ldots,c_{n-1},1$.
The successor rule
\begin{equation}\label{equ:rule1}
\rho_0 (c_0,c_1,\ldots,c_{n-1})=
\begin{cases}
\overline{c_0} & \mbox{if } wt(c_1,\ldots,c_{n-1})<\frac{n-1}{2} \mbox{ and } {\bf u}_{\bf c} \mbox{ is a necklace},\\
c_0 & \mbox{otherwise,}
\end{cases}
\end{equation}
joins all cycles in $T_0$, \ie, all cycles of PCR with weight less than $\frac{n+1}{2}$, into one cycle.
\end{lemma}

\begin{proof}
We consider the adjacency graph of cycles in $T_0$ and conjugate pairs between cycles in it.

Suppose that $C\neq({0}^n)$ is a nonzero cycle in $T_0$, the necklace of $C$ must be unique and ends with a $1$. Hence  there is only one state $1,c_1,\ldots,c_{n-1}$ in $C=(1,c_1,\ldots,c_{n-1})$ among all the states of $C$ satisfying the first condition of $\rho_0$, \ie,  $wt(c_1,\ldots,c_{n-1})<\frac{n-1}{2}$ and $c_1,\ldots,c_{n-1},1$ is necklace.
Then the conjugate state of $1,c_1,\ldots,c_{n-1}$ is $0,c_1,\ldots,c_{n-1}$, which is in another cycle $C'$ with $wt(C')=wt(C)-1$.   Hence there is only one conjugate pair $({\bf v},\hat{\bf v})$ such that ${\bf v}\in C$ and $wt({\hat{\bf v}})=wt(C)-1$.



Now we view the cycles in $T_0$ as vertices and conjugate pairs determined by the first condition of $\rho_0$ as edges  in the adjacency graph.  For a conjugate pair $({\bf v},\hat{\bf v})$, ${\bf v}\in C_1$ and $\hat{\bf v}\in C_2$, if $wt({\bf v})>wt({\hat{\bf v}})$, then we say there is an edge from $C_1$ to $C_2$.  From the earlier proof,  for arbitrary nonzero cycle $C_1$, there is one unique edge from $C_1$ to some $C_2$ such that $wt(C_2)=wt(C_1)-1$, and if $C_2\neq ({0}^n)$ there is one unique edge from $C_2$ to some $C_3$ such that $wt(C_3)=wt(C_2)-1$. Continue this process,  we can reach to cycle $({0}^n)$.   The adjacency graph is in fact a spanning tree with root $({ 0}^n)$. Therefore we can obtain a large cycle by joining all the cycles in $T_0$.
\end{proof}


Now we consider how to join the cycles in $T_1$.  

\begin{lemma}\label{lem2}
For each state ${\bf c}=c_0,c_1,\ldots,c_{n-1}$ such that $wt({\bf c})\geq \frac{n+1}{2}$, we let ${\bf u}_{\bf c}=\overline{c_{n-1}},\overline{c_{n-2}},\ldots,\overline{c_{1}},1$.
The successor rule
\begin{equation}\label{equ:rule1}
\rho_1 (c_0,c_1,\ldots,c_{n-1}) =
\begin{cases}
\overline{c_0} & \mbox{if } wt(c_1,\ldots,c_{n-1})>\frac{n-1}{2} \mbox{ and } {\bf u}_{\bf c} \mbox{ is a necklace},\\
c_0 & \mbox{otherwise,}
\end{cases}
\end{equation}
joins all cycles in $T_1$ into one cycle.
\end{lemma}

\begin{proof}
Similar to the proof of Lemma \ref{lem1}, we just need to prove that the cycles in $T_1$ and conjugate pairs determined by the first condition in $\rho_1$ form a spanning tree with root $({1}^n)$. It is enough to prove that  
for arbitrary cycle $C\neq ({1}^n)$ with $wt(C)\geq\frac{n+1}{2}$, the first condition in $\rho_1$ determines a unique conjugate pair $({\bf v},\hat{\bf v})$ satisfying ${\bf v}\in C$, $\hat{\bf v}\in C'$ and $wt(C')=wt(C)+1$.


For any cycle  $C$  in $T_1$ and not $({1}^n)$, we write
$C=(0,c_1,\ldots,c_{n-1})$ such  that  ${\bf u}=\overline{c_{n-1}},\overline{c_{n-2}},\ldots,\overline{c_{1}},1$ is necklace and $wt(c_1,\ldots,c_{n-1})>\frac{n-1}{2}$.
Let $C_1=({\bf u})$.
Obviously $C_1\in T_0$. Because of the uniqueness of necklace ${\bf u}$ in $C_1$, we deduce that $0,c_1,\ldots,c_{n-1}=cr{\bf u}$ is the unique state in $C$ which satisfies the first condition in $\rho_1$ and its conjugate state is $1,c_1,\ldots,c_{n-1}$ in a cycle $C'$ in $T_1$ with  the weight $wt(C)+1$.
\end{proof}

We remark that the successor rule in Lemma~\ref{lem2} without restriction on their weights can generate a de Bruijn sequence by complementing and reversing the sequence
 derived from the PCR3 successor.

\begin{lemma}\label{lem2a}
The successor
\begin{equation}\label{equ:rule1}
\rho_2 (c_0,c_1,\ldots,c_{n-1}) =
\begin{cases}
\overline{c_0} & \mbox{if }   (\overline{c_{n-1}},\overline{c_{n-2}},\ldots,\overline{c_{1}},1)  \mbox{ is a necklace},\\
c_0 & \mbox{otherwise,}
\end{cases}
\end{equation}
is a de Bruijn successor.
\end{lemma}

Next we study the relation between the two large cycles  generated by successor rules $\rho_0$ and $\rho_1$ respectively.

\begin{lemma}\label{lem3}
Let ${\bf s}$ and ${\bf t}$ be two sequences generated by $\rho_0$ and $\rho_1$ respectively. Then the period of ${\bf s}$ and ${\bf t}$ are all $2^{n-1}$. Furthermore, ${\bf s}$ and $\textbf{\emph{cr}}{\bf t}$ are shift equivalent.
\end{lemma}

\begin{proof}
The period of ${\bf s}$ and ${\bf t}$ are all $2^{n-1}$  because $n$ is odd and there are $2^{n-1}$ states with weight less than or equal to $\frac{n-1}{2}$ and larger than or equal to $\frac{n+1}{2}$ respectively.

We use induction to prove ${\bf s}$ and $\textbf{\emph{cr}}{\bf t}$ are shift equivalent. Here we use ${\bf t}_{-i}$ to represent the $(2^{n-1}-i)$-th state of ${\bf t}$, \ie, ${\bf t}_{2^{n-1}-i}$.

Without loss of generality, we let the initial state of ${\bf s}$ and ${\bf t}$ be ${0}^n$ and ${1}^n$ respectively.  Obviously,  ${0}^n=\textbf{\emph{cr}}{1}^n$. According to the successor rules $\rho_0$ and $\rho_1$, the successor of ${0}^n$ is ${0}^{n-1}1$ and the predecessor of ${1}^n$ is $0{1}^{n-1}$ satisfying
\[
{\bf s}_1={0}^{n-1}1=\textbf{\emph{cr}}(0{1}^{n-1})=\textbf{\emph{cr}}{\bf t}_{-1}.
\]
Suppose that for $i=0,1,2,\ldots,$ we have
\[
{\bf s}_i=s_0,s_1,\ldots,s_{n-1}=\textbf{\emph{cr}}{\bf t}_{-i}=\textbf{\emph{cr}}(t_0,t_1,\ldots,t_{n-1})=\overline{t_{n-1}},\overline{ t_{n-2}},\ldots,\overline{ t_0}.
\]

Then we consider the successor of ${\bf s}_i$ in two cases.  First,  if $s_1,\ldots,s_{n-1},1$ is a necklace, then the successor is ${\bf s}_{i+1}=s_1,\ldots,s_{n-1},\overline{s_0}$.
On the other hand, the predecessor of ${\bf t}_{-i}$ is ${\bf t}_{-i-1}=x,t_0,t_1,\ldots,t_{n-2}$ such that
\[
\overline{ t_{n-2}},\ldots,\overline{ t_0},1=s_1,\ldots,s_{n-1},1
\]
is also a necklace.  Hence  its successor must be $t_0,t_1,\ldots,t_{n-2},\bar{x}={\bf t}_{-i}$ and thus $x=\overline{t_{n-1}}$. So
\[
{\bf s}_{i+1}=s_1,\ldots,s_{n-1},\overline{s_0}=\textbf{\emph{cr}}{\bf t}_{-i-1}=cr (\overline{t_{n-1}},t_0,t_1,\ldots,t_{n-2})=\overline{t_{n-2}},\ldots,\overline{ t_0},t_{n-1}.
\]

Second, if $s_1,\ldots,s_{n-1},1$ is not a necklace, then we can similarly deduce that ${\bf s}_{i+1}=\textbf{\emph{cr}}{\bf t}_{-i-1}$ and the proof is complete.
\end{proof}

Finally we choose conjugate pairs of specific format from these two cycles generated by successor rules $\rho_0$ and $\rho_1$,  in order to  obtain a desired de Bruijn sequence.

\begin{lemma}\label{lem4}
By interchanging the successors of a conjugate pair $({\bf v},\hat{\bf v})$, where ${\bf v}=0,c_1,\ldots,c_{n-1}$ and $c_i=\overline{c_{n-i}}$, $i=1,2,\ldots,\frac{n-1}{2}$, two cycles generated by $\rho_0$ and $\rho_1$ respectively can be joined into an order $n$  de Bruijn sequence ${\bf s}$ satisfying ${\bf s}=\textbf{\emph{cr}}{\bf s}$.
\end{lemma}

\begin{proof}
Suppose ${\bf s}$ is the de Bruijn sequence  generated  with initial state ${\bf v}$.   Then we have
\[
{\bf s}=(0,c_1,\ldots,c_{n-1},1,\cdots) \mbox{ and }  L^{n+1}{\bf s}=(\cdots, 0,c_1,\ldots,c_{n-1},1).
\]
Using Lemma \ref{lem3},  it is easy to check that ${\bf s}=\textbf{\emph{cr}}(L^{n+1}{\bf s})$.
\end{proof}

Because there are $2^{\frac{n-1}{2}}$ distinct ${\bf v}=0,c_1,\ldots,c_{n-1}$ satisfying $c_i=\overline{c_{n-i}}$, $i=1,2,\ldots,\frac{n-1}{2}$,  we can obtain $2^{\frac{n-1}{2}}$ inequivalent CR de Bruijn sequences, from the cycles generated by $\rho_0$ and $\rho_1$.

Summarizing the results in Lemmas \ref{lem1} - \ref{lem4}, we derive a formal successor rule to generate CR de Bruijn sequence.

\begin{theorem}\label{th1}
For each state ${\bf c}=c_0,c_1,\ldots,c_{n-1}$, we let ${\bf u}_{\bf c}=c_1,c_2,\ldots,c_{n-1},1$ if $wt(c_1,c_2,\ldots,c_{n-1})< \frac{n-1}{2}$ and
${\bf u}_{\bf c}=\overline{c_{n-1}},\overline{c_{n-2}},\ldots,\overline{c_{1}},1$ if $wt(c_1,c_2,\ldots,c_{n-1})> \frac{n-1}{2}$. Let ${\bf d}$ be any vector in $B$, where
\[
B=\{(c_1,\ldots,c_{n-1})\in\mathbb{F}_2^{n-1} \mid c_i=\overline{c_{n-i}}, i=1,2,\ldots,\frac{n-1}{2}  \}.
\]
The successor rule
\begin{equation}\label{equ:rule}
\rho (c_0,c_1,\ldots,c_{n-1})=
\begin{cases}
\overline{c_0} & \mbox{if } {\bf u}_{\bf c} \mbox{ is a necklace},\\
\overline{c_0} & \mbox{if } (c_1,c_2,\ldots,c_{n-1})={\bf d},\\
c_0 & \mbox{otherwise,}
\end{cases}
\end{equation}
generates a CR de Bruijn sequence.
\end{theorem}

We consider the complexity of this successor rule. It is clear that the space complexity is $O(n)$. Checking whether ${\bf u}_{\bf c}$ is a necklace can be done in $O(n)$ time which was established in~\cite{Gabric18}. Thus, the successor rule $\rho$ in Theorem \ref{th1} requires time and space complexities $O(n)$ to generate the next bit of the de Bruijn sequence.

As a consequence, we formulate the following statement which solves the problem proposed by Fredricksen \cite{Fred82}.

\begin{theorem}
When $n>1$ is odd, there always exist order $n$ de Bruijn sequences satisfying that the complement and reverse sequences are the same.
\end{theorem}

In the sequel, we provide a few concrete examples.

\begin{example} \label{examples}
When $n=5$, $B=\{0011,1100,0101,1010\}$. If we take different vectors in $B$ as ${\bf d}$ in successor rule $\rho$ in Theorem \ref{th1}, we can get 4 inequivalent CR de Bruijn sequences (see Table \ref{table:1}).

\begin{table*}[h!]
	\caption{Inequivalent de Bruijn sequences constructed based on Theorems~\ref{th1} with $n=5$.}
	\label{table:1}
	\renewcommand{\arraystretch}{1.2}
	\centering
	{\footnotesize
		\begin{tabular}{ccc}
			\hline
			Entry & ${\bf d}=d_1d_2d_3d_4$ & Resulting sequence based on Theorem~\ref{th1}\\
			\hline
			$1$   & $0011$                   & $(1	1	1	1	1	0	1	1	0	1	0	1	1	1	0	0	1	1	0	0	0	1	0	1	0	0	1	0	0	0	0	0)$\\
			$2$   & $1100$                 & $(	0	0	1	1	1	1	1	0	1	1	0	1	0	1	1	1	0	0	0	1	0	1	0	0	1	0	0	0	0	0 1	1)$\\
			$3$   & $0101$                 & $(0	1	0	0	1	0	0	0	0	0 1	1	0	0	0	1	0	1	1	1	0	0	1	1	1	1	1	0	1	1	0	1	)$\\
			$4$   & $1010$                 & $(1	0	0	1	0	0	0	0	0 1	1	0	0	0	1	0	1	0	1	1	1	0	0	1	1	1	1	1	0	1	1	0	)$\\
						\hline
					\end{tabular}
	}
\end{table*}

When $n=7$, the successor rule $\rho$ in Theorem \ref{th1} can generate 8 distinct CR de Bruijn sequences with period 128.  For example,  if we take ${\bf d}=000111$,  then the resulting sequence is
\begin{align*}
{\bf s}= & (1111111011101001110110110011011110101101010111001011111001111000\\
         &  1110000110000010110001010100101000010011001001000110100010000000).
\end{align*}

When $n=9$, the successor rule $\rho$ in Theorem \ref{th1} can generate $2^{\frac{9-1}{2}}=16$ distinct CR de Bruijn sequences with period 512.  Again, if we take ${\bf d}=00001111$, then the resulting sequence is
\begin{align*}
{\bf s}=(&1111111110111101110111001110111110110110111010110110010110111100\\
         &1101101001101110001101111110101110110101100110101111010101101010\\
         &1011100101011111001011101001011110001011111110011110110011100110\\
         &0111110100111010100111100100111111000111011000111101000111110000\\
         &1111000001110100001110010001110000001101100001101010001101000001\\
         &1001100011001000011000000010111000010110100010110000010101100010\\
         &1010100101010000101001100101001000101000000100111000100110100100\\
         &1100001001011001001010001001001000001000110001000100001000000000).
\end{align*}
\end{example}

We note that there are other ways to join all the cycles in $T_0$ and join all cycles in $T_1$  in order to generate CR de Bruijn sequences. For example, we can use PCR4 successor in \cite{Gabric18} combined with the weight restriction as we did earlier. 
\begin{theorem}
For each state ${\bf c}=c_0,c_1,\ldots,c_{n-1}$, we let
${\bf u}_{\bf c}=0,c_1,c_2,\ldots,c_{n-1}$ if $wt(c_1,c_2,\ldots,c_{n-1})>\frac{n-1}{2}$ and
${\bf u}_{\bf c}=0,\overline{c_{n-1}},\overline{c_{n-2}},\ldots,\overline{c_{1}}$ if $wt(c_1,c_2,\ldots,c_{n-1})< \frac{n-1}{2}$. Let ${\bf d}$ be any given vector in $B$,
where
\[
B=\{(c_1,\ldots,c_{n-1})\in\mathbb{F}_2^{n-1} \mid c_i=\overline{c_{n-i}}, i=1,2,\ldots,\frac{n-1}{2}  \}.
\]

The successor rule
\begin{equation}\label{equ:rule}
\rho (c_0,c_1,\ldots,c_{n-1})=
\begin{cases}
\overline{c_0} & \mbox{if } {\bf u}_{\bf c} \mbox{ is a necklace},\\
\overline{c_0} & \mbox{if } (c_1,c_2,\ldots,c_{n-1})={\bf d},\\
c_0 & \mbox{otherwise,}
\end{cases}
\end{equation}
generates a CR de Bruijn sequence.
\end{theorem}

Using weights of cycles as a criterion to  divide the cycles of PCR into two parts  makes  it easier to  propose successor rules to fast generate CR de Bruijn sequences. In fact there are other ways of dividing cycles of PCR  in order to generate such de Bruijn sequences. For example, when $n=5$, PCR provides 8 distinct cycles. We can divide theses cycles into two parts
\[
T_0=\{(0^5),(00001),(00011),(01011)\},
T_1=\{(1^5),(01111),(00111),(00101)\},
\]
satisfying that if cycle $C$ is in $T_0$ then $\textbf{\emph{cr}}C$ is in $T_1$.
We can join the cycles in $T_0$ by using conjugate pairs $(0^5,10^4)$, $(01000,11000)$ and $(00110,10110)$, into one cycle $S=(0000010001101011)$. It is easy to check that $\textbf{\emph{cr}}S$ can be obtained by joining cycles in $T_1$.
Using special conjugate pair $(01100,11100)$, we can join $S$ and $\textbf{\emph{cr}}S$ into a CR de Bruijn sequence
\[
{\bf s}=(0000010001101011,0010100111011111).
\]
However, it is a little bit complicated to describe the corresponding successor rule to efficiently generate ${\bf s}$.


\section{Distribution  of  de Bruijn sequences}\label{Sec-properties}


In this section, we refine the characterization of CR de Bruijn sequences and use it to study the number of CR de Bruijn sequences, as well as the distribution of  de Bruijn sequences with the maximum linear complexity.

Suppose that ${\bf s}$ is an  order $n$ de Bruijn sequence satisfying $\textbf{\emph{c}}{\bf s}=\textbf{\emph{r}}{\bf s}$. According to Lemma \ref{lemma1}, ${\bf s}=(S\textbf{\emph{cr}}S)$ for some $S$ with period $2^{n-1}$  satisfying that 
 $\textbf{\emph{cr}}{\bf v}$ isn't in $S$ if  the state ${\bf v}$ is in $S$.
Without loss of generality, we can assume that the all zeroes state $0^n$ is in $S$.
 Let $S=(s_0,s_1,\ldots,s_{N-1})$, where $N=2^{n-1}$.  Then
\begin{equation}\label{equ:cr-1}
{\bf s}=(S\textbf{\emph{cr}}S)=(s_0,s_1,\ldots,s_{N-1},\overline{s_{N-1}},\ldots,\overline{s_1},\overline{s_0}).
\end{equation}
In this case, we first show that the CR de Bruijn sequence ${\bf s} = (S\textbf{\emph{cr}}S) $ that is presented by concatenating $S$ and $\textbf{\emph{cr}}S$
can be obtained by the CJM  from joining $S =(s_0,s_1,\ldots,s_{N-1}) $ and $\textbf{\emph{cr}}S = (\overline{s_{N-1}},\ldots,\overline{s_1},\overline{s_0})$ according to a specific conjugate pair.

\begin{lemma}\label{lem31}
Suppose that ${\bf s}$ is a CR de Bruijn sequence with the form in Equation (\ref{equ:cr-1}) such that  the all zeroes state $0^n$ is in $S$. Then
\begin{enumerate}
  \item ${\bf s}$ can be written as
  \begin{align*}
  (&s_0,\ldots,s_{\frac{n-3}{2}},s_{\frac{n-1}{2}}, s_{\frac{n+1}{2}}, \ldots,  s_{N-\frac{n+3}{2}},  s_{\frac{n-1}{2}},\overline{s_0},\ldots,\overline{s_{\frac{n-3}{2}}},\\
   &s_{\frac{n-3}{2}},\ldots,s_0,\overline{s_{\frac{n-1}{2}}},   \overline{s_{N-\frac{n+3}{2}}}, \ldots,  \overline{s_{\frac{n+1}{2}}},  \overline{s_{\frac{n-1}{2}}},,\overline{s_{\frac{n-3}{2}}},\ldots,\overline{s_0}).
  \end{align*}
  \item The conjugate pair $({\bf v},\hat{\bf v})$ such that  ${\bf v}$ in $S$,  which  is used to join $S$ and $\textbf{\emph{cr}}S$ into one cycle ${\bf s}$,  satisfies the form
  \[
  {\bf v}=s_{\frac{n-1}{2}},\overline{s_0},\ldots,\overline{s_{\frac{n-3}{2}}},s_0,\ldots,s_{\frac{n-3}{2}}.
  \]
  \item $S$ is unique.
\end{enumerate}
\end{lemma}

\begin{proof}
Let us consider which conjugate pair can be used to join $S$ and $\textbf{\emph{cr}}S$ to obtain ${\bf s}$. We note that  the period of $S$  and
$\textbf{\emph{cr}}S$ are both $N$.
If one of the conjugate states in $S$ is $s_{N-n},\ldots,s_{N-1}$,
then its conjugate state in $\textbf{\emph{cr}}S$ must be $\overline{s_{n-1}},\ldots,\overline{s_0}$ satisfying
\[
s_{N-n}=s_{n-1}, s_{N-n+1}=\overline{s_{n-2}},\ldots, s_{N-1}=\overline{s_0}.
\]
Hence in $S$ there exists consecutive terms
\[
s_{n-1},\overline{s_{n-2}},\ldots,\overline{s_0},s_0,\ldots,s_{n-2},s_{n-1}.
\]
Obviously, both the state   $(\overline{s_0},s_0,\ldots,s_{n-2})$  and
\[
\textbf{\emph{cr}}(\overline{s_0},s_0,\ldots,s_{n-2}) = (\overline{s_{n-2}},\ldots,\overline{s_0},s_0)
\]
are in $S$,  which contradicts to that  $(S\textbf{\emph{cr}}S)$ is a CR de Bruijn sequence.

For other cases, suppose that one of the conjugate states in $S$ is
\[
s_{N-t},\ldots,s_{N-1},s_0,\ldots,s_{n-t-1}, \ \ 1\leq t\leq n-1.
\]
Then its conjugate state in $\textbf{\emph{cr}}S$ must be $\overline{s_{t-1}},\ldots,\overline{s_0},\overline{s_{N-1}},\ldots,\overline{s_{N-n+t}}$ satisfying
\begin{equation}\label{equ:ccc}
s_{N-t}={s_{t-1}},s_{N-t+1}=\overline{s_{t-2}},\ldots, s_{N-1}=\overline{s_0},s_0=\overline{s_{N-1}},\ldots,s_{n-t-1}=\overline{s_{N-n+t}}.
\end{equation}
If $t>\frac{n+1}{2}$, then in $S$ there exists consecutive terms
\[
s_{t-1},\overline{s_{t-2}},\ldots,\overline{s_0},s_0,\ldots,s_{t-1}
\]
with length $2t>n+1$.  This implies that $n-t<t-1$ and thus the state ${\bf v}=(\overline{s_{t-2}},\ldots,\overline{s_0},s_0,\ldots,s_{n-t})$  and the state
\[
\textbf{\emph{cr}}{\bf v} = (\overline{s_{n-t}},\ldots,\overline{s_0},s_0,\ldots,s_{t-2})
\]
are both in $S$, a contradiction.
If $t<\frac{n+1}{2}$, then  $t -1 < n-t$ and thus  from  Equation (\ref{equ:ccc}) we obtain both $s_{N-t}={s_{t-1}}$ and ${s_{t-1}} = \overline{s_{N-t}}$ simultaneously, which is also impossible.

Hence the only possible conjugate pair  ${\bf v}$ in $S$ and  $\hat{\bf v}$ in $\textbf{\emph{cr}}S$ that can be used  to join them into a CR de Bruijn sequence is
\[
{\bf v}=s_{N-\frac{n+1}{2}},\ldots,s_{N-1},s_0,\ldots,s_{\frac{n-3}{2}}
\]
and
\[
\hat{\bf v}=\overline{s_{\frac{n-1}{2}}},\ldots,\overline{s_0},\overline{s_{N-1}},\ldots,\overline{s_{N-\frac{n-1}{2}}}
\]
satisfying
\[
s_{N-\frac{n+1}{2}}=s_{\frac{n-1}{2}},s_{N-\frac{n-1}{2}}=\overline{s_{\frac{n-3}{2}}},\ldots,s_{N-1}=\overline{s_0}.
\]
Indeed, by Lemma~\ref{lem4}, we can use this conjugate pair to join $S$ and $\textbf{\emph{cr}}S$ into ${\bf s}$.  The  first $\frac{n-1}{2}$ terms of $S$ are  $s_0,\ldots,s_{\frac{n-3}{2}}$ and the last
$\frac{n+1}{2}$ terms of $S$ are  $s_{\frac{n-1}{2}},\overline{s_0},\ldots,\overline{s_{\frac{n-3}{2}}}$.  The rest of the proof follows immediately.
\end{proof}

Hence each CR de Bruijn sequence ${\bf s}$  corresponds a  unique $S$ containing the all zero state and a conjugate pair that can be used to join $S$ and $\textbf{\emph{cr}}S$ through the concatenation.

\begin{lemma}\label{lem32}
For any given $S$ such that $(S\textbf{\emph{cr}}S)$ is a CR de Bruijn sequence,  let ${\bf v}$ be any state  in $S$  satisfying
\begin{equation*}
{\bf v}=c_0,c_1,\ldots,c_{n-1}, \ \ \ \  c_i=\overline{c_{n-i}}, i=1,2,\ldots,\frac{n-1}{2}.
\end{equation*}
If ${\bf v}$ in $S$, then the conjugate state $\hat{\bf v}$ is in $\textbf{\emph{cr}}S$. Furthermore, the next bit of $0, c_1,\ldots,c_{n-1}$ is 0 and the next bit of $1, c_1,\ldots,c_{n-1}$ is 1.
\end{lemma}

\begin{proof}
If ${\bf v}$ and $\hat{\bf v}$ are both in $S$, then two successors $c_1,\ldots,c_{n-1},c_0$ and $c_1,\ldots,c_{n-1},\overline{c_0}$ of them must be in $S$. Then
\[
cr(c_1,\ldots,c_{n-1},\overline{c_0})=c_0,\overline{c_{n-1}},\ldots,\overline{c_1}=c_0,c_1,\ldots,c_{n-1}={\bf v}
\]
is also in $\textbf{\emph{cr}}S$ and thus $(S\textbf{\emph{cr}}S)$ is not a de Bruijn sequence, a contradiction. Similarly,  ${\bf v}$ and $\hat{\bf v}$ can't be both in $\textbf{\emph{cr}}S$.

Suppose that $0, c_1,\ldots,c_{n-1}$ is in $S$ and the next bit is 1. In this case,  $\textbf{\emph{cr}}S$ also contains
\[
cr(c_1,\ldots,c_{n-1},1)=0,\overline{c_{n-1}},\ldots,\overline{c_1}=0, c_1,\ldots,c_{n-1},
\]
which contradicts to that $(S\textbf{\emph{cr}}S)$ is a de Bruijn sequence. The rest of the proof is similar.
\end{proof}

Now, let us take any given $S$ such that $(S\textbf{\emph{cr}}S)$ is a CR de Bruijn sequence.  Let $({\bf v},\hat{\bf v})$ be any conjugate pair  in $S$ and $\textbf{\emph{cr}}S$ satisfying
\begin{equation}\label{equ:v}
{\bf v}=c_0,c_1,\ldots,c_{n-1}, \ \ \ \  c_i=\overline{c_{n-i}}, i=1,2,\ldots,\frac{n-1}{2}.
\end{equation}

By Lemma~\ref{lem31},  we can use  any such conjugate pair to locate a left shift of  $S$ and a right shift of $\textbf{\emph{cr}}S$  to generate another CR de Bruijn sequence  $(L^t(S)  \textbf{\emph{cr}}L^t(S))=(L^t(S) R^t( \textbf{\emph{cr}}(S)))$.  Different conjugate pairs correspond to different  CR de Bruijn sequences.  This is equivalent to say that we can use CJM to generate  $2^{\frac{n-1}{2}}$ distinct CR de Bruijn sequences using these different conjugate pairs.

\begin{lemma}\label{lem34}
For a given $S$ such that ${\bf s}=(S\textbf{\emph{cr}}S)$ is a CR de Bruijn sequence, we can generate $2^{\frac{n-1}{2}}$ distinct CR de Bruijn sequences using conjugate pairs $({\bf v},\hat{\bf v})$ where ${\bf v}$ has the form in Equation (\ref{equ:v}).
\end{lemma}

We note that if ${\bf s}=(S\textbf{\emph{cr}}S)$ is CR de Bruijn sequence, then $\textbf{\emph{r}}S$ also contains all zeroes state and
\[
((\textbf{\emph{r}}S)\textbf{\emph{cr}}(\textbf{\emph{r}}S))=(\textbf{\emph{r}}S\textbf{\emph{c}}S)=(\textbf{\emph{c}}S\textbf{\emph{r}}S)=\textbf{\emph{r}}{\bf s}
\]
is also a CR de Bruijn sequence.
When $n=3$, there are two CR de Bruijn sequences: ${\bf s}=(10001110)$ and $r{\bf s}=(00010111)$. In particular,  $S=(0001)$ and $ \textbf{\emph{r}}S=(1000)=(0001)=S$.
But for $n\geq 5$, $\textbf{\emph{r}}S$ is not equivalent to $S$.

\begin{lemma}\label{lem35}
Suppose that $n\geq 5$ and ${\bf s}=(S \textbf{\emph{cr}}S)$ is a CR de Bruijn sequence. Then $\textbf{\emph{r}}S$ is inequivalent to $S$.
\end{lemma}
\begin{proof}
When $n\geq 5$ is odd, there are $2^{\frac{n+1}{2}}$ symmetric states $a_0,a_1,\ldots,a_{n-1}$ in ${\bf s}=(S\textbf{\emph{cr}}S)$ satisfying $a_i=a_{n-1-i}, i=0,1,\ldots,\frac{n-1}{2}-1$.
Because $\textbf{\emph{r}}(a_0,a_1,\ldots,a_{n-1})$ is still symmetric,  there are $2^{\frac{n-1}{2}}\geq 4$ symmetric states in $S$.
If $S=\textbf{\emph{r}}S$, then $S$ must be symmetric, which is impossible because in $S$ there exist only two symmetric positions: the beginning (or end) and the midpoint (see  \cite{Mayhew94}). Hence the proof is complete.
\end{proof}

Some operators on de Bruijn sequences are defined in \cite{Etzion84}. 
For a de Bruijn sequence ${\bf s}$, let $\textbf{\emph{z}}{\bf s}$ (resp. $\textbf{\emph{u}}{\bf s}$) denote the sequence obtained from ${\bf s}$, by interchanging the positions of the unique runs of $n$ and $n-2$ zeroes (resp. ones). One can check that
$\textbf{\emph{z}}{\bf s}$ and $\textbf{\emph{u}}{\bf s}$ are de Bruijn sequences. Furthermore,
 let $\textbf{\emph{d}}{\bf s}$ denote the sequence obtained from ${\bf s}$ by deleting two zeroes and two ones from the unique runs of $n$ zeroes and $n$ ones respectively.
For example, for ${\bf s}=(0000111101100101)$, we obtain
\begin{align*}
\textbf{\emph{z}}{\bf s}&=(0011110110000101),\\
\textbf{\emph{u}}{\bf s}&=(0000110111100101),\\
\textbf{\emph{uz}}{\bf s}=\textbf{\emph{zu}}{\bf s}&=(0011011110000101),\\
\textbf{\emph{d}}{\bf s}&=(001101100101).
\end{align*}

Here we extend the operator $\textbf{\emph{d}}$ to general sequence $S$ and let $\textbf{\emph{d}}S$ denote the sequence obtained from $S$ by deleting two zeroes and two ones from the unique runs of $n$ zeroes and $n$ ones respectively when these two runs exist.

According to Lemma 4 in \cite{Etzion84}, we can deduce the following result directly.
\begin{lemma}\label{lem333}
Suppose that $n\geq 5$ and ${\bf s}=(S\textbf{\emph{cr}}S)$ is a CR de Bruijn sequence.
\begin{enumerate}
  \item ${\bf s}'=\textbf{\emph{uz}}{\bf s}=(S'\textbf{\emph{cr}}S')$ for some $S'$ containing state $0^n$ is also a CR de Bruijn sequence.
  \item $\textbf{\emph{d}}{\bf s}=\textbf{\emph{d}}{\bf s}'$ is a $CR$-sequence.
  \item $\textbf{\emph{d}}S'$ is equal to $\textbf{\emph{d}}S$ or $\textbf{\emph{dcr}}S$.
  \item $S'$ is determined from $S$ uniquely and vice versa.
  \item $S$, $\textbf{\emph{r}}S$, $S'$ and $\textbf{\emph{r}}S'$ are pairwise inequivalent.
\end{enumerate}
\end{lemma}
\begin{proof}
  If $S$ contains the run of  $n-2$ zeros, then  ${\bf s'}$ is obtained from ${\bf s}$ by interchanging the positions of the runs of  $n$ zeros and $n-2$ zeros in $S$
and the positions of the runs of $n$ ones and $n-2$ one in $\textbf{\emph{dcr}}S$. In this case,  $\textbf{\emph{d}}S' = \textbf{\emph{d}}S$ and $S'$ is obtained from $S$
by interchanging the positions of the runs of  $n$ zeros and $n-2$ zeros in $S$.
Otherwise,  ${\bf s'}$ is obtained from ${\bf s}$ by interchanging the positions of the runs of  $n$ zeros in $S$ and $n-2$ zeros in $\textbf{\emph{dcr}}S$,
and the positions of the runs of $n$ ones  in $\textbf{\emph{dcr}}S$ and $n-2$ ones in $S$.  In this case, $\textbf{\emph{d}}S' = \textbf{\emph{dcr}}S$ and $S'$ is obtained from
$\textbf{\emph{cr}}S$ by deleting two ones in the run of $n$ ones and inserting two zeros in the run of $n-2$ zeros.  The rest of proof follows immediately.
\end{proof}

Combining the above lemmas, we obtain the following result on the number of de Bruijn sequences being $CR$-sequences.

\begin{theorem}
The number of CR de Bruijn sequences is a multiple of   $2^{\frac{n+3}{2}}$ for odd order $n\geq5$.
\end{theorem}

\begin{proof}
By    Lemma \ref{lem31}, the CR de Bruijn sequence ${\bf s}$   can be written as $(S\textbf{\emph{cr}}S)$ for a unique $S$ containing all zeroes state. For such a given $S$ satisfying $(S\textbf{\emph{cr}}S)$ being CR de Bruijn sequence, by Lemma \ref{lem34},
we can generate $2^{\frac{n-1}{2}}$ distinct CR de Bruijn sequences. Furthermore, by Lemma \ref{lem333},   $S$, $\textbf{\emph{r}}S$, $S'$ and $\textbf{\emph{r}}S'$ are pairwise inequivalent when $n\geq5$.   Therefore,  if there is a CR de Bruijn sequence ${\bf s}=(S\textbf{\emph{cr}}S)$, then we can generate exact $4\times2^{\frac{n-1}{2}}=2^{\frac{n+3}{2}}$ distinct CR de Bruijn sequences.
\end{proof}

Now let us discuss the linear complexities of CR de Bruijn sequences. Generally speaking, joining $S$ and $crS$ using distinct conjugate pairs will generate CR de Bruijn sequences with distinct linear complexities. For example, let $n=5$ and $S=(1011101100000110)$. This example was first considered in \cite{Fred82}.  Using two distinct conjugate pairs $({\bf v}_1,\hat{\bf v}_1)$ and $({\bf v}_2,\hat{\bf v}_2)$ where ${\bf v}_1=11010$ and ${\bf v}_2=00011$ respectively,  we obtain the resulting sequences
\[
(1011101100000110,1001111100100010)
\]
and
\[
(1101011101100000,1111100100010100)
\]
with linear complexities 23 and 29 respectively. However,  we show next  the CR de Bruijn sequences obtained  by joining $S$ and $\textbf{\emph{cr}}S$ as in Lemma~\ref{lem34} will have the same linear complexity  when ${\bf s}=(S\textbf{\emph{cr}}S)$ has  the maximal linear complexity $2^n-1$.

\begin{lemma}\label{lem36}
Suppose that $n\geq5$ and ${\bf s}=(S \textbf{\emph{cr}}S)$ is a CR de Bruijn sequence with linear complexity $2^n-1$, then all the CR de Bruijn sequences generated by joining $S$ and $\textbf{\emph{cr}}S$ have the same linear complexity.
\end{lemma}

\begin{proof}
By  Lemma 1 in \cite{Etzion84}, a binary sequence ${\bf a}=(a_0,a_1,\ldots,a_{2^n-1})$ with period $2^n$ has the linear complexity $2^n-1$ if and only if the subparity
\[
sp({\bf a})=a_0+a_2+\cdots+a_{2^n-2}=a_1+a_3+\cdots+a_{2^n-1}= 1.
\]
Suppose that $S=(s_0,s_1,\ldots,s_{2^{n-1}-1})$ and
\[
{\bf s}=(S\textbf{\emph{cr}}S)=(s_0,s_1,\ldots,s_{2^{n-1}-1},\overline{s_{2^{n-1}-1}},\ldots,\overline{s_1},\overline{s_0})
\]
is a CR de Bruijn sequence with linear complexity $2^n-1$.  Then
\[
sp({\bf s})=\sum_{i=0}^{2^{n-2}-1}s_{2i}+\sum_{i=0}^{2^{n-2}-1}\overline{s_{2i+1}}=\sum_{i=0}^{2^{n-2}-1}s_{2i+1}+\sum_{i=0}^{2^{n-2}-1}\overline{s_{2i}}=\sum_{i=0}^{2^{n-1}-1}s_i+2^{n-2}=1,
\]
{\it i.e.}, $s_0+s_1+\cdots+s_{2^{n-1}-1}=1$.

If we join $S$ and $\textbf{\emph{cr}}S$ to generate another CR de Bruijn sequence ${\bf s}'$, then it must have the following form:
\[
(L^t(S)\textbf{\emph{cr}}L^t(S))=(s_t,\ldots,s_{2^{n-1}-1},s_0,\ldots,s_{t-1},\overline{s_{t-1}},\ldots,\overline{s_0},\overline{s_{2^{n-1}-1}},\ldots,\overline{s_t}),
\]
for some $t$, and
\[
sp({\bf s}')=\sum_{i=0}^{2^{n-2}-1}s_{t+2i}+\sum_{i=0}^{2^{n-2}-1}\overline{s_{t-1-2i}}=\sum_{i=0}^{2^{n-1}-1}s_i+2^{n-2}=1,
\]
hence the linear complexity of ${\bf s}'$ is also $2^n-1$.
\end{proof}

It is a classical problem to study the distribution of   de Bruijn sequences with given linear complexity.   Let  $\gamma(c,n)$ be the number of order $n$ de Bruijn sequences with linear complexity $c$.  It is well known that  $\gamma(c,n) \equiv 0~({\rm mod}\ 2)$ for $n\geq3$  and  $\gamma(c,n) \equiv 0~({\rm mod}\ 4)$ for even $n\geq 4$. The latter works because
${\bf s}$,  $\textbf{\emph{c}}{\bf s}$, $\textbf{\emph{r}}{\bf s}$, $\textbf{\emph{cr}}{\bf s}$ are pairwise inequivalent de Bruijn sequences of the same complexity for even $n$.

In 1982, Chan et al. \cite{Chan82} conjectured  $\gamma(c,n) \equiv 0~({\rm mod}\ 4)$.  In \cite{Etzion84} Etzion and Lempel proved the absence of CR de Bruijn sequences of even complexity and thus confirmed this conjecture for all $c$ and $n$ such that $cn$ is even. However,  the argument does not work when $cn$ is odd and this conjecture is still open.

On the other hand, when the linear complexity is the maximum $2^n-1$, Etzion and Lempel proved the following stronger result.

\begin{theorem}\cite[Theorems 1 and 2]{Etzion84}\label{thm:222}
For $n \geq 4$,  $\gamma(2^k-1, n)\equiv 0~({\rm mod}\ 8)$.
For $k\geq 3$, $\gamma(2^k-1,2k)\equiv 0~({\rm mod}\ 16)$.
\end{theorem}

These results are derived using the fact  ${\bf s}, {\bf cs}, {\bf rs}$ have the same linear complexity 
and the following lemmas.

\begin{lemma} \cite[Lemmas 2, 4]{Etzion84}
Let $G_1=\{\textbf{\emph{e}},\textbf{\emph{r}},\textbf{\emph{z}},\textbf{\emph{u}},\textbf{\emph{rz}},\textbf{\emph{ru}},\textbf{\emph{zu}},\textbf{\emph{rzu}}\}$, where $\textbf{\emph{e}}$ is the identity operator. For $n\geq 5$ and for each de Bruijn sequence ${\bf s}$, $G_1{\bf s}$ consists of eight pairwise inequivalent de Bruijn sequences.
Furthermore, if the linear complexity of ${\bf s}$ is $2^n-1$, then sequences in $G_1{\bf s}$ have the same linear complexity $2^n-1$.
\end{lemma}
\begin{lemma} \cite[Lemma 7]{Etzion84} \label{lem7}
 If ${\bf s}$ is a de Bruijn sequence of order $2k$, $k\geq 3$, then the union of $G_1{\bf s}$ and $G_1\textbf{\emph{c}}{\bf s}$ consists of sixteen pairwise inequivalent de Bruijn sequences.
\end{lemma}

Based on  previous properties of CR de Bruijn sequences of odd order,  we can extend Theorem~\ref{thm:222} to the following.

\begin{theorem} \label{distribution}
For $n\geq 5$, $\gamma(2^n-1,n) \equiv 0~({\rm mod}\ 16)$.
\end{theorem}

\begin{proof}
We only need to prove this result holds for odd $n\geq5$. Suppose that ${\bf s}$ is a de Bruijn sequence of odd order $n\geq5$.
According to the proof of Lemma 7 in \cite{Etzion84},  $\textbf{\emph{dc}}{\bf s}\neq \textbf{\emph{d}}{\bf s}$ still holds.  Moreover,
if $\textbf{\emph{dc}}{\bf s}\neq \textbf{\emph{dr}}{\bf s}$,  then we can still generate 16 pairwise inequivalent de Bruijn sequences. Furthermore,  sequences in $G_1{\bf s}$
and $G_1{\bf cs}$ have the same linear complexity  if ${\bf s}$ has the maximum linear complexity $2^n-1$. Therefore,
these 16 de Bruijn sequences have the same linear complexity.

If $\textbf{\emph{dc}}{\bf s}= \textbf{\emph{dr}}{\bf s}$, {\it i.e.}, $\textbf{\emph{d}}{\bf s}= \textbf{\emph{dcr}}{\bf s}$,
then there exists a CR de Bruijn sequence in $G_1{\bf s}$ and $G_1{\bf s} = G_1\textbf{\emph{cr}}{\bf s} = G_1\textbf{\emph{c}}{\bf s}$. In this case,  we can not use Lemma~\ref{lem7}.
Without loss of generality, we can assume that ${\bf s}$ is a CR de Bruijn sequence with the form $(S\textbf{\emph{cr}}S)$ for some $S$. According to Lemma \ref{lem34}, we can use the given $S$ to generate $2^\frac{n-1}{2}$ distinct CR de Bruijn sequences of the form ${\bf s}'=(L^t(S)\textbf{\emph{cr}}L^t(S))$, for some $t$.  Moreover,  $G_1{\bf s}$ and $G_1{\bf s}'$ are disjoint because ${\bf s}'$ and $\textbf{\emph{r}}{\bf s}'$ are not in $G_1{\bf s}$; indeed, distinct conjugate pairs are used  to join $S$ and $\textbf{\emph{cr}}S$ to generate these distinct CR de Bruijn sequences.

Hence, for odd $n\geq 5$ and any given  CR de Bruijn sequence ${\bf s}=(S\textbf{\emph{cr}}S)$, we can generate
\[
2^{\frac{n-1}{2}}\times |G_1{\bf s}|=2^{\frac{n-1}{2}}\times2^3=2^{\frac{n+5}{2}}> 2^4
\]
distinct de Bruijn sequences. If the linear complexity of ${\bf s}$ is $2^n-1$, then by Lemma \ref{lem36}, all these sequences have the linear complexity $2^n-1$. The proof is complete.
\end{proof}

CR de Bruijn sequences with  maximum linear complexity  do exist.  For example,  one can verify easily that the CR de Bruijn sequences of order $n=5$ and $n=9$ generated by our earlier algorithm in Section~\ref{sec:main} (see Example~\ref{examples}) have maximum linear complexities.   From the proof of Theorem~\ref{distribution}, we have  $32$ distinct  de Bruijn sequences of order $5$ for each given CR de Bruijn sequence of order $5$ with maximum linear complexity,  and $128$  distinct  de Bruijn sequences of order $9$ similarly.

\section{Conclusions}\label{Sec4}

In this paper, we solved a problem by Fredricksen saying that for any odd $n>1$ there exists an order $n$ de Bruijn sequence which is a $CR$-sequence, \ie, its complement sequence and reverse sequence are the same. Our proof is constructive, which uses the classical cycle joining method based on partitions of cycle structure of PCR successor according to their weights.  We provided an efficient successor rule to generate many de Bruijn sequences with the desired property.
Similar successor rules can also be deduced using other classical de Bruijn successor rules, together with the partition of cycle structures.
Moreover, we refined a characterization of CR de Bruijn sequences and showed  the number of CR  de Bruijn sequences is a multiple of $2^{\frac{n+3}{2}}$ for odd $n\geq 5$.  Finally we studied the distribution of order $n$ de Bruijn sequences with the maximum linear complexity and proved these numbers are congruent to $0$ modulo $16$.


\end{document}